\documentclass[conference]{IEEEtran}
\usepackage{authblk}
\usepackage{graphicx}
\usepackage[tbtags]{amsmath}
\usepackage{amssymb, amsthm, bbm}
\usepackage{float}
\usepackage{gensymb}
\usepackage{color}
\newtheorem{theorem}{Theorem}
\newtheorem{lemma}{Lemma}

\title{Closed-form analysis  of Multi-RIS Reflected Signals in RIS-Aided Networks Using Stochastic Geometry }

\author[1, 2]{Guodong~Sun}
\author[1, 2]{Fran\c{c}ois~Baccelli}
\affil[1]{Department d'informatique, Ecole Normale Sup\'erieure. Email: guodong.sun$|$francois.baccelli@ens.fr}
\affil[2]{Institut national de recherche en sciences et technologies du numérique (INRIA)}

\begin{document}
\maketitle
\begin{abstract}
Reconfigurable intelligent surfaces (RISs) enhance wireless communication by creating engineered signal reflection paths in addition to direct links. 
This work presents a stochastic geometry framework using point processes (PPs) to model multiple randomly deployed RISs conditioned on their associated base station (BS) locations.
By characterizing aggregated reflections from multiple RISs using the Laplace transform, we analytically assess the performance impact of RIS-reflected signals by integrating this characterization into well-established stochastic geometry frameworks.
Specifically, we derive closed-form expressions for the Laplace transform of the reflected signal power in several deployment scenarios.
These analytical results facilitate performance evaluation of RIS-enabled enhancements. 
Numerical simulations validate that optimal RIS placement favors proximity to BSs or user equipment (UEs), and further quantify the impact of reflected interference, various fading assumptions, and diverse spatial deployment strategies. 
Importantly, our analytical approach shows superior computational efficiency compared to Monte Carlo simulations.

\end{abstract}

\section{Introduction}
The reconfigurable intelligent surface (RIS) technology, with its dynamic control of reflected waves, provides a promising solution for improving cellular network performance, particularly for coverage extension and data rate enhancement~\cite{liu2022path, wu2019intelligent}.
However, the irregular placement of RISs and the signal fluctuation of their reflection links pose significant challenges for performance analysis, especially multiple RISs~\cite{di2019reflection}.
To address this challenge, stochastic geometry has been applied for modeling and evaluating the performance of cellular networks at a system level~\cite{baccelli2010stochastic, baccelli2020random}. 

A common approach is to approximate the composite signal formed by the direct and reflected links using a representative distribution for a specific RIS location, and subsequently consider the spatial distribution of this RIS.
For example, the authors in \cite{lyu2021hybrid} model both base stations (BSs) and RISs as homogeneous Poisson point processes (PPP), and associate the typical user equipment (UE) to the nearest BS and RIS. 
Given an association layout, the composite signal power is approximated by a Gamma distribution using the moment matching method, and then a double integral expression is derived to quantify the coverage probability considering the nearest neighbor distributions of both the BS and the RIS. 
To consider the spatial correlation between RISs and BSs, the authors in~\cite{wang2023performance} use a Gauss-Poisson process to model the BSs and RISs with a fixed separation distance. They also apply the moment matching framework to approximate a Gamma distributed composite signal for performance analysis. 

Recent works have proposed analytical solutions for the RIS reflected signals instead of the moment matching approximation. 
For instance, the authors in \cite{deng2024modeling} accurately characterize the composite signal as Erlang-distributed when both the direct and RIS-reflected links experience Rayleigh fading conditioned on the RIS location. 
Then, the ergodic rate is obtained by averaging over the spatial distribution of this RIS.  
Similar to the previous moment matching method, merging the reflected signal into the direct link becomes computationally intractable with multiple RISs, as it necessitates integrating over the location of each RIS.
In contrast, the authors in \cite{sun2023performance} proposed a framework that models multiple RISs as a point process (PP),  characterizing the aggregated reflected signal power by its Laplace transform. 
This framework accommodates general fading and arbitrary RIS deployment areas, offering the key advantage of analytical tractability regardless of the number of RISs, which is reviewed and utilized as the general framework in this paper. 

While analytical expressions in integral form have been developed for characterizing the reflected signals from RISs, closed-form solutions are most desirable for efficient stochastic geometry analysis.
For example, the authors in \cite{andrews2011tractable} derived a closed-form expression for the Laplace transform of aggregated interference under Rayleigh fading with a path loss exponent of 4, providing essential insights for understanding the system interference.
The contribution of this work is the derivation of such closed-form expressions for the Laplace transform of the aggregated reflection signals, detailed in Section~\ref{section:special_closed_form}. 
These expressions apply to the scenario where multiple RISs are positioned at a fixed distance from their associated BS, with a path loss exponent of 2 and under both Rayleigh and Nakagami-$m$ fading.
{ Our closed-form analysis focuses on the reflected signal, which is sufficient for signal-to-noise ratio (SNR)-based analysis of multi-RIS assisted networks. 
This focus is justified by existing BS interference models~\cite{andrews2011tractable} and our planned exploration of RIS reflected interference with flexible assumptions. }

This work is organized as follows: 
Section~\ref{section:system_model} introduces the stochastic geometry system model for the RIS-assisted cellular networks.
Section~\ref{section:general_soluation} reviews the general analytical framework for performance evaluation in multi-RIS assisted networks. 
Closed-form expressions for the Laplace transforms of the aggregated reflection signal power are derived in Section~\ref{section:special_closed_form}.
Then, Section~\ref{sec:simulation} presents numerical results that demonstrate the impact of RIS modeling and deployment strategies on system performance evaluation.
Finally, Section~\ref{sec:conclusion} concludes the paper and discusses future research directions.

\section{System Model}\label{section:system_model}
\begin{figure}
\centering
\includegraphics[width=\linewidth]{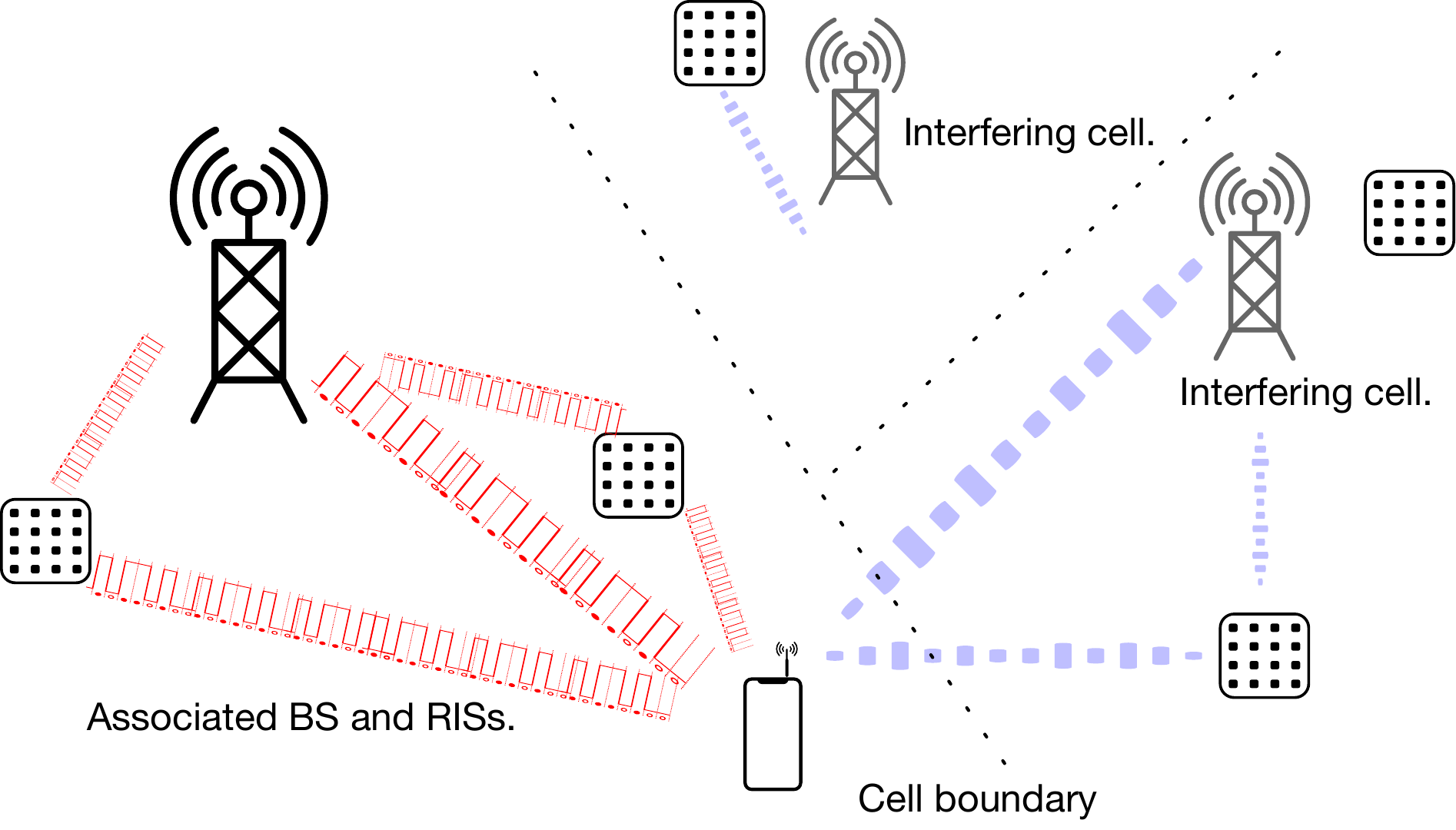}
\caption{Randomly located RISs reflect signals to assist BS-UE communication, with inter-cell BSs and RISs generating interference. }
\label{fig:system_model}
\vspace{-0.5cm}
\end{figure}
To model the spatial randomness of a RIS-assisted cellular network, we first represent the location of BSs as a homogeneous PPP, denoted by 
$\Phi_{\rm BS} \triangleq \sum_{i\in \mathcal I} \delta_{\mathbf{x}_i}$
with intensity $\lambda_{{\rm BS}}$, where $\mathcal{I}$ is the index set of all BSs and $\delta$ is the Dirac measure. 
We assume that each cell manages its own RISs.
Therefore, we model RIS locations as the union of conditionally independent PPs $\Phi_{\rm RIS} \triangleq \cup_{i\in \mathcal{I}} \phi_i$,
where $\phi_i\triangleq \sum_{j \in \mathcal{J}_i}\delta_{\mathbf{y}_{i,j}}$ represents the RIS locations managed by the cell of BS $i$ and $\mathcal{J}_i$ is the index set of these RISs.
The specific distribution of the RIS PP is intentionally left unspecified at this stage, as Section~\ref{section:general_soluation} will present a general solution applicable to various RIS spatial distributions. 
Specialized PPs will be introduced and discussed in Section~\ref{section:special_closed_form}. 

Next, we model the RIS reflection. 
We assume that each RIS performs beamforming by manipulating the phase of signals reflected by RIS elements to direct the reflected beam towards its associated UE. 
When aligned in phase, these signals superpose in amplitude, resulting in an amplitude gain proportional to the number of RIS elements, denoted by $M$. 
Consequently, the RIS reflection power gain scales quadratically with $M$, resulting in a gain of $M^2$.
In neighboring cells, RISs are configured by their respective BSs to reflect beamformed interference towards their target UE. 
This interference may be directed and hence overlap with the tagged UE. 
We assume an interference overlap probability of this random overlap, denoted by $\xi$, can be either fixed or a function of the beamwidth.
Specifically, we model the beamwidth as being inversely proportional to the number of RIS elements $M$, which leads to an overlap probability of $\xi =C_{\rm beam}/M$, where $C_{\rm beam}$ is a constant characterizing the beamforming capacity. 
We neglect the effects of RISs not configured for beamforming, such as those in neighboring cells that randomly scatter signals, as their contributions are already captured within the environmental fading model.

We analyze RIS-assisted downlink transmission from a BS to a tagged UE, where the UE is located in origin and at the nearest neighbor distance $r = \|\mathbf{x}_o\|$ from its serving BS. Here, the index $o$ refers to this serving BS. 
Due to the spatial distribution of RISs, the varying propagation distances of direct and reflected paths result in a time dispersive composite channel between the BS and UE.
As demonstrated in \cite{sun2023performance}, orthogonal frequency-division multiplexing (OFDM) modulation can mitigate the time-dispersion caused by multipath reflections, and the post-processed signal-to-interference-plus-noise ratio (SINR) is given by 
\vspace{-0.1cm}
\begin{equation}\label{eq:SINR-def}
\text{SINR} = \frac{|\gamma_{D_o}|^2 + \sum_{j \in \mathcal{J}_o} | \gamma_{{R}_{o,j}}|^2 }{\sum_{i\in \mathcal{I}\setminus o}( |\gamma_{D_i}|^2 + \sum_{j \in \mathcal{J}_i} \mathbbm{1}_{\{\xi_{i,j}\}} | \gamma_{{R}_{i,j}}|^2  ) + \sigma^2},
\vspace{-0.1cm}
\end{equation}
where $\sigma^2$ is the power of the additive Gaussian noise.
Here, we use $\gamma_{D_o}$ and $\gamma_{R_{o,j}}$ to denote the received signal amplitude of the direct and reflected intended signals, respectively, while $\gamma_{D_i}$ and $\gamma_{{R}_{i,j}}$ represent the corresponding received signal amplitude of the direct and reflected interference, respectively.
Each signal amplitude is the product of: a fading variable $\rho$, the square root of the transmission power $\sqrt{P_0}$, a path loss function dependent on the propagation distance, and, for the reflected links, a beamforming gain $M$.
Specifically, the direct signal amplitude is given by $ \gamma_D = \rho_{D} \sqrt{g(d)P_0}$ for distance $d$, while the reflected signal amplitude is $ \gamma_R = \rho_{R} M \sqrt{g(d_1)g(d_2) P_0}$, where the latter accounts for the two segments of the reflected path, i.e., $d_1=\|\mathbf{y}_{i,j}-\mathbf{x}_i \|$ and $d_2 =\|\mathbf{x}_i \| $. 
Furthermore, the indicator function $\mathbbm{1}_{\{\xi_{i,j}\}}$  represents the event that the interfering beam from the $j^{\rm th}$ RIS in the $i^{\rm th}$ cell overlaps with the tagged UE.
Below, we assume a power-law path loss function $g(d) = \beta d^{-\alpha}$, where $\alpha$ denotes the path loss exponent.
Here, the parameter $\beta = \big(\frac{c}{4\pi f_c}\big)^2$ represents the average power gain at a reference distance of $1$ meter, with $f_c$ being the carrier frequency and $c$ the speed of light. 

We assume that the direct links experience Rayleigh fading and are subject to a non line of sight (NLoS) channel. 
For reflected links, we will employ fading distributions consistent with existing literature, such as non-central Gaussian distributed fading~\cite{sun2023performance}, Rayleigh fading~\cite{deng2024modeling}, and Nakagami-$m$ fading~\cite{badiu2019communication, qian2020beamforming}. 
We also consider channel hardening~\cite{bjornson2020rayleigh}, where the reflected link's signal amplitude is assumed to be constant.
Specifically, we will utilize Rayleigh and Nakagami-$m$ fading in Section~\ref{section:special_closed_form} to derive closed-form expressions.

\section{General Framework}\label{section:general_soluation}

In this section, we first review the established stochastic geometry framework for wireless network performance analysis.
Subsequently, we analyze the Laplace transform of the aggregated reflected signal power from multiple RISs.
Finally, we examine how the deployment of multi-RISs impacts performance analysis and how these new network elements can be integrated into the existing framework. 

\begin{figure}
	\centering
	\includegraphics[width=\linewidth]{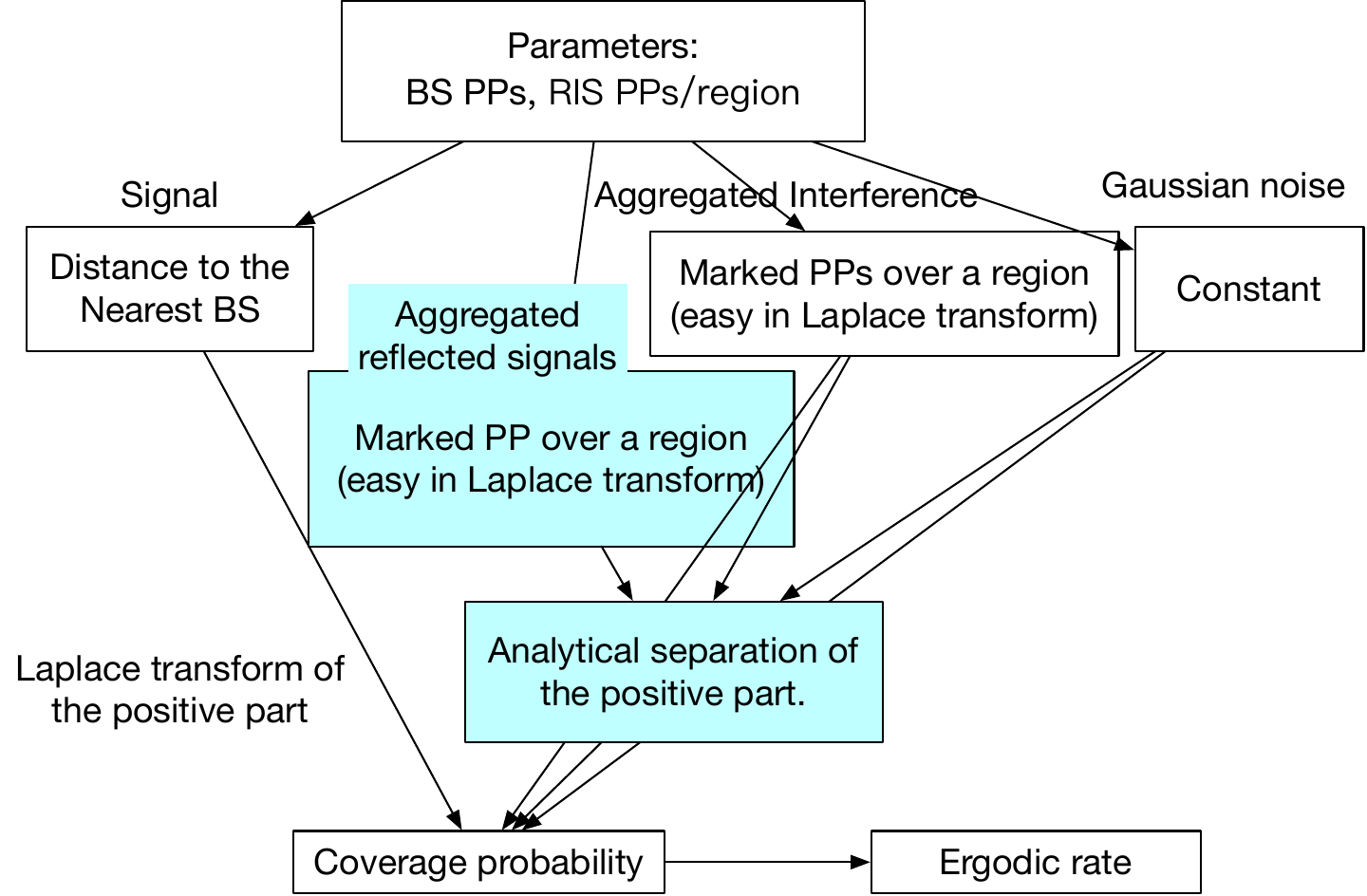}
	\caption{A stochastic geometry-based performance analysis framework for  cellular networks (uncolored), including RIS enhancement (colored).}
	\label{fig:workflow}
	\vspace{-0.5cm}
\end{figure}

\subsection{Classical stochastic geometry analysis}
The uncolored part of Fig.~\ref{fig:workflow} illustrates the workflow of cellular network performance analysis using stochastic geometry, where the locations of the BSs are modeled as a PPP~\cite{andrews2011tractable, andrews2016primer}.
This framework treats the received power of the desired signal and the aggregated interference as random variables. 
{\color{black} To enable convenient manipulation of expressions, we will consistently use $Q$ to represent the power of a given signal type in this work. 
Recall that $r$ is the BS-UE distance. We denote the desired signal power as  $Q_{S}(r) =|\rho_{D}|^2 P_0 g(r)$, and the aggregated interference power from independent interfering BSs as  $Q_I(r) = \sum_{i\neq o} |\rho_{D_i}|^2 P_0 g(\| \mathbf{x}_i\|)$.}
We have 
\vspace{-0.1cm}
\begin{equation}\label{eq:SINR_components_no_RIS}
    \mbox{SINR} = \frac{ Q_S(r)}{Q_{I}(r)+\sigma^2 }.
    \vspace{-0.1cm} 
\end{equation}
The desired signal power, $Q_{S}(r)$, has a fading power $|\rho_{D}|^2$ that follows an exponential distribution and a signal attenuation of $P_0 g(r)$.
The aggregate interference power, $Q_I(r)$, generated from independent interfering BSs whose location is a PPP defined in the area away from the associated BS at distance $r$, has a Laplace transform given by
\vspace{-0.1cm}
\begin{equation}\label{eq:interference_classical}
	\mathcal{L}_{Q_I(r)}(s) = e^{-2\pi\lambda_{\rm BS}\int_r^{\infty}x\big(1- \mathcal{L}_{v(x)}(s)\big)\text{d}x},
\vspace{-0.1cm}
\end{equation}
where {\color{black} $v(x)$ denotes the interference power from a cell whose BS at distance $x$; and $\mathcal{L}_{v(x)}(s)$ is its Laplace transform, given by $\mathcal{L}_{v(x)}(s) = \mathcal{L}_{|\rho_{D}|^2}(s) =\frac{1}{1+sP_0 g(x)}$ since the interference of a cell is solely from the BS and experience Rayleigh fading}. 
Finally, let $\sigma^2$ denote the Gaussian noise power with Laplace transform $\mathcal{L}_{ \sigma^2}(s) = e^{-s\sigma^2}$.

The coverage probability, $\mathsf{P}_{c}(T)$ , is defined as the probability that the SINR of the tagged UE exceeds a target threshold $T > 0$. 
Conditioned on the distance $r$ from the associated BS to the tagged UE, the corresponding coverage probability is $\mathsf{P}_{c}(T|r)\triangleq \mathbb{P}\left(\mbox{SINR} \geq T |r\right)$.
This can be analytically characterized by~\cite{andrews2011tractable}
\vspace{-0.1cm}
\begin{equation}\label{eq:classical_laplace}
\begin{aligned}[b]
&  \mathbb{P}\left[|\rho_{D}|^2  P_0 g(r) \geq T( Q_I(r) +\sigma^2 )\Big|r\right] \\
=&   \mathbb{P}\left[|\rho_{D}|^2\geq \frac{ T(Q_I(r) +\sigma^2)}{ P_0 g(r)}\Big| r \right] \\
\stackrel{(a)}{=} & \int_{0}^{\infty} e^{-\frac{\upsilon T}{ P_0g(r)}} f_{ Q_I(r) +\sigma^2}(\upsilon){\rm d}\upsilon \\
\stackrel{(b)}{=}&\mathcal{L}_{ Q_I(r) }\left(\frac{T}{P_0g(r)}\right)\mathcal{L}_{ \sigma^2}\left(\frac{T}{P_0g(r)}\right),
 \end{aligned}
\end{equation}
where $(a)$ results from the complementary cumulative distribution function (CCDF) of exponential distribution, defined as  $\overline{F}_{|\rho_{D}|^2}(t)=e^{-t}$; $(b)$ follows from the definition of unilateral Laplace transform. 
Here, $f_{ Q_I(r) +\sigma^2}$ denotes the probability density function (PDF) of the random variable $ Q_I(r) +\sigma^2$.

The ergodic rate can be obtained from the coverage probability $\mathsf{P}_c(T|r)$, given by~\cite{baccelli2010stochastic}
\vspace{-0.1cm}
\begin{equation}\label{eq:ergodic_rate}
\tau(r) \triangleq \mathbb{E}[\log(1+\text{SINR})|r] = \int_{0}^{\infty} \frac{\mathsf{P}_c(t|r)}{t+1}{\rm d}t.
\vspace{-0.1cm}
\end{equation}

\subsection{Laplace Transform of reflected signals}
To analyze the aggregated signal or the interference power reflected by multiple RISs within a cell, denoted by $Q_c$, we model their deployment as a marked PP where each point (RIS location) is marked by its independent reflected power $|\rho_{R}|^2M^2P_0g(d_1)g(d_2)$. 
This allows us to derive the distribution of the aggregated reflection power using the Laplace transform, given the RIS deployment region and individual link fading characteristics.
For instance, when modeling the RIS deployment as a homogeneous PPP over a region $\mathbb{X}$ of arbitrary shape,  the Laplace transform is expressed as
\begin{equation}\label{eq:interference_RIS_modification}
	\mathcal{L}_{Q_c}(s) = e^{-\lambda_{\rm RIS} \int_{\mathbb{X}} (1-\mathcal{L}_{|\rho_{R}|^2(x)}(sM^2P_0g(d_1)g(d_2))) {\rm d}x }, 
\end{equation}
where $\mathcal{L}_{|\rho_{R}|^2(x)}(s)$ is the Laplace transform of the reflected signal power from a RIS located at $x\in \mathbb{X}$. 
Alternative PPs, such as the binomial point process (for a fixed number of RISs) or the Ginibre point process (for repulsive RIS deployments)~\cite{kong2016exact}, can also be applied to the analysis described below.

Generally, the Laplace transform of the aggregated reflected signal power is expressed as an integral over the RIS deployment region.
In Section~\ref{section:special_closed_form},
 we present a special scenario that enables the derivation of closed-form expressions for this Laplace transform.

\subsection{Framework adaption for RIS deployment}

With the deployment of RISs, the desired signal, $Q_{S}$, comprises two components as given in Eq.~\eqref{eq:SINR-def}: the direct signal power $Q_{S_D}(r) = |\rho_{D}|^2 P_0 g(r)$ and the reflected signal power $Q_{S_R}(r) = \sum_{j\in \mathcal{J}_o} |\rho_{{R_{o,j}}}|^2 P_0 G(\mathbf{x}_o, \mathbf{y}_{o,j})$. 
Then, the SINR expression is
\vspace{-0.2cm}
\begin{equation}\label{eq:SINR_w_RIS}
    \mbox{SINR} = \frac{ Q_{S_D}(r)+Q_{S_R}(r)}{Q_{I}(r)+\sigma^2 }. 
\vspace{-0.1cm}
\end{equation}
To include the RIS contributed interference into $Q_{I}(r)$ in Eq.~\eqref{eq:interference_classical}, we simply modify the Laplace transform of the interference power from a cell $\mathcal{L}_{v(x)}(s)$ into
\vspace{-0.2cm}
\begin{equation}
	\mathcal{L}_{v(x)}(s) = \frac{\mathcal{L}_{Q_{I_c(x)}}(s)}{1+sP_0g(x)},
\vspace{-0.1cm}
\end{equation}
where $\mathcal{L}_{Q_{I_{c}(x)}}(s)$ represents the Laplace transform of the RIS-reflected interference from the RISs in the cell whose serving BS is at a distance $x$, which can be adapted when the RIS PP is specified. 
Since all the components in Eq.~\eqref{eq:SINR_w_RIS} and the coverage threshold are non-negative, the coverage probability $\mathsf{P}_{c}(T|r)$ in the scenario with RISs is obtained by manipulating the SINR in terms of the following expression:
\begin{equation}\label{eq:cov_expression}
\begin{aligned}
\mathbb{P}\big[Q_{S_D}(r)\geq T(Q_I(r)+\sigma^2)-Q_{S_R}(r)|r\big]. 
\end{aligned}
\end{equation}
Due to the subtraction of the reflected signal $-Q_{S_R}(r)$, the right-hand side is no longer strictly non-negative. 
The extension from a non-negative to a real-valued right-hand side is addressed within the colored portion of Fig.~\ref{fig:workflow} and presented below.
We denote
\vspace{-0.2cm}
\begin{equation}\label{eq:definition}
    \Upsilon = T\big(Q_I(r)+\sigma^2 \big)-Q_{S_R}(r).
\vspace{-0.1cm}
\end{equation} 
The bilateral Laplace transform of $\Upsilon$ is given by
\vspace{-0.1cm}
\begin{equation}\label{eq:bilater_laplace_product}
\begin{aligned}
		\mathcal{B}_{\Upsilon}(s) =&  \mathcal{B}_{Q_I(r)}(sT)\mathcal{B}_{\sigma^2}(sT)\mathcal{B}_{-Q_{S_R}(r)}(s) \\
		=&\mathcal{L}_{Q_I(r)}(sT)\mathcal{L}_{\sigma^2}(sT)\mathcal{L}_{-Q_{S_R}(r)}(s),
\end{aligned}
\vspace{-0.1cm}
\end{equation}
where the last equality holds because the the random variables $Q_I(r)$, $\sigma^2$, and $Q_{S_R}(r)$ are defined only on the non-negative real axis. Here, $\mathcal{L}_{-Q_{S_R}(r)}(s)$ will be defined in Section~\ref{section:special_closed_form}.
Given the Laplace transform $\mathcal{B}_{\Upsilon}(s)$, we have
\begin{lemma}\label{lemma:coverage}
When RISs are configured as beamformers, the coverage probability for the communication coverage threshold $T$ is given by
\vspace{-0.2cm}
\begin{equation}\label{eq:prop1}
\begin{aligned}
\mathsf{P}_c (T|r) =  \mathcal{B}_{ \Upsilon^+ }\left(\frac{1}{P_0g(r)}\right),
\end{aligned}
\vspace{-0.2cm}
\end{equation}
where $ \Upsilon^+ \triangleq \max\{0, \Upsilon\}$. 
\end{lemma}
\begin{proof}
From Eq.~\eqref{eq:cov_expression} and the definition $Q_{S_D} = |\rho_{D}|^2  P_0 g(r)$, the coverage probability $\mathsf{P}_c (T|r)$ is given by
\vspace{-0.2cm}
\begin{equation}\label{eq:positive}
	\begin{aligned}[b]
&  \mathbb{P}\left[|\rho_{D}|^2  P_0 g(r) \geq \Upsilon\Big|r\right]
=   \mathbb{P}\left[|\rho_{D}|^2\geq \frac{\Upsilon}{ P_0 g(r)}\Big| r \right] \\
\stackrel{(a)}{=} & \int_{0}^{\infty} e^{-\frac{\upsilon}{ P_0g(r)}} f_{\Upsilon}(\upsilon){\rm d}\upsilon + \int_{-\infty}^{0}f_{\Upsilon}(\upsilon){\rm d}\upsilon \\
 \stackrel{(b)}{=}&  \mathcal{B}_{ \Upsilon^+ }\left(\frac{1}{P_0g(r)}\right),
\end{aligned}
\vspace{-0.2cm}
\end{equation}
where $(a)$: the exponential term in the first integral is derived from the CCDF of the exponential distribution. For this integral, the variable $\Upsilon$ is integrated over the positive part for the first integral since the signal fading power $|\rho_D|^2$ is defined on the $\mathbb{R}^+$.  
$(b)$ follows from the definition of the bilateral Laplace transform, where $\mathcal{B}_{\Upsilon^+}(s) = \int_{0}^{\infty}e^{-s\upsilon}f_{\Upsilon}(\upsilon){\rm d}\upsilon + \int_{-\infty}^{0}e^{-s0}f_{\Upsilon}(\upsilon){\rm d}\upsilon$ (since $e^{-s0}=1$ for the second integral), with the argument $s=$ set to $\frac{1}{P_0g(r)}$ .
\vspace{-0.2cm}
\end{proof}
The difference in coverage probability expressions between the RIS-assisted scenario and BS-only scenario lies in the negative component of $\Upsilon$, which is introduced by the RIS reflection signals $-Q_{S_R}$. 
Specifically, when $\upsilon$ takes a negative value (as seen in the second integral of the expression $(a)$ of Eq.~\eqref{eq:positive}), the communication coverage is guaranteed by the reflected link, resulting in a coverage probability of 1 regardless of direct link fading. 
This contrasts with the coverage probability in Eq.~\eqref{eq:classical_laplace}, where $\Upsilon$ only takes positive values.

To connect $\mathcal{B}_{\Upsilon}(s)$ defined in Eq.~\eqref{eq:bilater_laplace_product} with the coverage probability of RIS-assisted networks $\mathcal{B}_{\Upsilon^+}(s)$, we have:
\begin{theorem}
The Laplace transform of the positive part of a random variable can be derived from its bilateral Laplace transform via the formula 
\vspace{-0.1cm}
\begin{equation}\label{eq:theorem}
\begin{aligned}
& \mathcal{B}_{\Upsilon^+}(s)= \\ 
& \frac{1}{2\pi \imath}\int_{-\infty}^{\infty} \Big(\mathcal{B}_{\Upsilon}(s-\imath u)- \mathcal{B}_{\Upsilon}(-\imath u) \Big)\frac{{\rm d} u}{u} + \frac{1}{2}\Big( 1+ \mathcal{B}_{\Upsilon}(s)\Big),
\end{aligned}
\end{equation}
where $\int_{-\infty}^{\infty}\frac{{\rm d}u}{u}$ is understood in the sense of Cauchy principal-value, that is $ \int_{-\infty}^{\infty}= \lim_{\epsilon \downarrow 0^+}\int_{-\infty}^{-\epsilon}+  \int_{\epsilon}^{\infty}$.
\end{theorem}
\begin{proof}
    Please refer to Appendix~D in~\cite{sun2023performance}. 
\end{proof}

Applying this framework necessitates several layered numerical integrations: integrating over the interference region, integrating over the RIS deployment region, and computing the Laplace transform for the positive component of the distribution.
To complement the analytical capabilities of this framework, we introduce a numerical solver that enables practical application by accommodating arbitrary deployment areas and link characteristics\footnote{https://github.com/Sun8Guden/RISassistedCellularFramework.git}.

\section{Special Case with Closed-Form Expressions}\label{section:special_closed_form}

\begin{figure}
\centering
	\includegraphics[width=\linewidth]{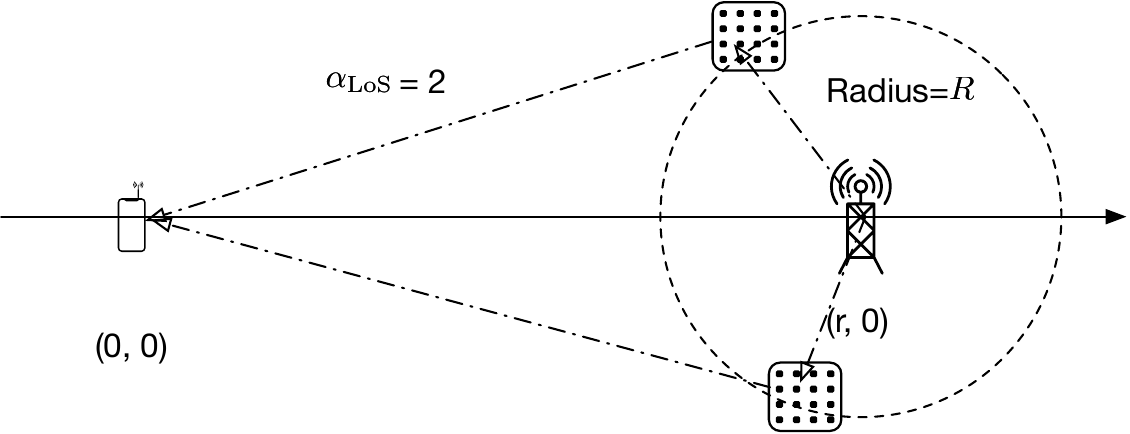}
	\caption{The RISs are located at a constant distance $R$ from the BS, with their angular distribution being random along a circle.}
\label{fig:special_scenario}
\vspace{-0.5cm}
\end{figure}

RISs offer optimal performance enhancement when placed in proximity to either BSs or UEs, a point that will be validated in Section~\ref{subsec:proximity}. 
In this section, we focus on the scenario where RISs are placed near BSs.
Despite random RIS placements within cells, constrained by proximity guidelines, spatial averaging across multiple cells yields a circular deployment pattern around BSs.
This pattern is characterized by an average BS-RIS distance of $R$, as shown in Fig.~\ref{fig:special_scenario}.
We assume that multiple RISs are placed along the dashed circle, with their spatial distributions modeled by PPPs for a random number of RISs, or BPPs for a fixed number of RISs.
We assume that the RIS-reflected links  experience either Rayleigh or Nakagami-$m$ fading, and the path loss exponent between RIS and UE, $\alpha_{\rm LoS}$, as 2.
This modeling allows us to characterize the Laplace transform of the RIS-reflected signal power with a closed-form expression.

\subsection{PPP and Rayleigh modeling}

We first derive a closed-form expression for the scenario where the location of RISs is modelled as a PPP with density $\lambda_{\rm RIS}$, and the reflected links experience Rayleigh fading. 
We have
\vspace{-0.2cm}
\begin{equation}\label{eq:laplace_reflect_PPP}
\begin{aligned}
	\mathcal{L}_{-Q_R}(s)&=
		 e^{-\lambda_{\rm RIS} \!\int_{0}^{2\pi}\!R\big(1-\mathcal{L}_{|\rho_R|^2}(-sM^2P_0G(r, R, \psi)) \big){\rm d}\psi } \\
	&\stackrel{(a)}{=} e^{-\lambda_{\rm RIS} R\int_{0}^{2\pi}\big(1-\frac{1}{1 - sM^2 \beta^2P_0 R^{-2}|r+Re^{-\imath \psi}|^{-2}} \big){\rm d}\psi }, \\
\end{aligned}
\end{equation}
where $(a)$ follows from $G(r, R, \psi) = \beta^2 R^{-2}|r+Re^{-\imath \psi}|^{-2}$ and $\mathcal{L}_{|\rho_R|^2}(s) = \frac{1}{1+s}$ for Rayleigh fading.
By substituting $s=\frac{1}{P_0g(r)}$ from Eq.~\eqref{eq:prop1} and $g(r)=\beta r^{-\alpha_{\rm NLoS}}$ into this integral, we have 
\vspace{-0.2cm} 
\begin{equation}\label{eq:reducible_integral}
\begin{aligned}[b]
	&\int_{0}^{2\pi}\bigg(1-\frac{1}{1 - s M^2\beta^2P_0 R^{-2}|r+Re^{-\imath \psi}|^{-2}} \bigg){\rm d}\psi  \\
	=&\int_{0}^{2\pi} \bigg(1 - \frac{1}{1-\frac{M^2\beta r^{\alpha_{\rm NLoS}}|r+Re^{-\imath \psi}|^{-2})}{R^2}}  \bigg) {\rm d}\psi \\
%	=& \int_{0}^{2\pi} \bigg(1 - \frac{|r+Re^{-\imath \psi}|^{2}}{|r+Re^{-\imath \psi}|^{2}-\frac{M^2\beta r^{\alpha_{\rm NLoS}}}{R^2}}  \bigg) {\rm d}\psi \\
	=& \int_{0}^{2\pi} \bigg( - \frac{\frac{M^2\beta r^{\alpha_{\rm NLoS}}}{R^2}}{|r+Re^{-\imath \psi}|^{2}-\frac{M^2\beta r^{\alpha_{\rm NLoS}}}{R^2}}  \bigg) {\rm d}\psi \\
	=&\int_{0}^{2\pi} \bigg( - \frac{\frac{M^2\beta r^{\alpha_{\rm NLoS}}}{R^4}}{| \frac{r}{R}+e^{-\imath \psi}|^2 - \frac{M^2\beta r^{\alpha_{\rm NLoS}}}{R^4}}\bigg) {\rm d}\psi.
\end{aligned}
\vspace{-0.1cm}
\end{equation}
Observe that Eq~\eqref{eq:reducible_integral} can be expressed in closed form. 
Let $\frac{r}{R} = a$, $\frac{M^2\beta r^{\alpha_{\rm NLoS}}}{R^4} = b$, we get
\vspace{-0.1cm}
\begin{equation}\label{eq:elimination}
\begin{aligned}[b]
	&\int_0^{2\pi} \bigg(-\frac{b}{-b+|a+e^{-\imath \psi}|^2}\bigg){\rm d}\psi \\
	=&\int_0^{2\pi} \bigg(-\frac{b}{-b+\big(a+\cos(\psi)\big)^2+\big(\sin(\psi)\big)^2}\bigg){\rm d}\psi \\
	=& -\frac{2\pi b}{\sqrt{a^4 + ( b - 1)^2 - 2 a^2 (1 + b)}}.
\end{aligned}
\vspace{-0.1cm}
\end{equation}
Plugging Eq.~\eqref{eq:elimination} into Eq.~\eqref{eq:laplace_reflect_PPP}, we have
\vspace{-0.2cm}
\begin{equation}\label{eq:closed_form_reflected_signal}
\begin{aligned}
	\mathcal{L}_{-Q_R}&\Big(\frac{1}{P_0g(r)}\Big)=
	\\
	&e^{\frac{ 2\pi \lambda_{\rm RIS} M^2 \beta r^{\alpha_{\rm NLoS}}/R^3}{\sqrt{(\frac{r}{R})^4 + (\frac{M^2 \beta r^{\alpha_{\rm NLoS}}}{R^4}-1)^2 - 2 (\frac{r}{R})^2 (1 + \frac{M^2 \beta r^{\alpha_{\rm NLoS}}}{R^4})}}}.
\end{aligned}
\vspace{-0.1cm}
\end{equation}
This closed-form Laplace transform can be directly incorporated into the general framework for performance metric derivation.

In our system model, RISs in other cells beamform and reflect signals from their associated BSs. 
These reflected beams may overlap the tagged UE with probability $\xi$.
Consequently,  the RIS reflection interference in each interfering cell can be modeled as a PPP thinned by $\xi$. Therefore, the Laplace transform of the aggregated interference from RISs, $\mathcal{L}_{Q_{I_c}(x)}(s)$, as defined in Eq.~\eqref{eq:interference_RIS_modification}, is given by
\vspace{-0.2cm}
\begin{equation}
\begin{aligned}[b]
	\mathcal{L}_{Q_{I_c}(x)}(s) = e^{-\lambda_{\rm RIS} \xi \int_{0}^{2\pi} R \big(1-\frac{1}{1+sP_0M^2\beta^2R^{-2}|x+Re^{-\imath \psi}|^{-\alpha_{IR}} }\big) {\rm d}\psi  },
\end{aligned}
\vspace{-0.1cm}
\end{equation}
where the path loss exponent for the interfering RIS-UE link $\alpha_{IR}$ may not be equal to 2, and the Laplace transform may not be reduced to closed form as in Eq~\eqref{eq:closed_form_reflected_signal}, thus requiring numerical integration. 

\subsection{Extension to BPP modeling and Nakagami-$m$ fading}

The closed-form Laplace transform for the reflected signal power, derived for PPP-distributed RISs and Rayleigh fading, represents the fundamental characteristics of this derivation. Extensions to a BPP model and Nakagami-$m$ fading are possible and discussed in this subsection.

For BPP-modeled RISs with $N$ RISs, we have the Laplace transform for the power of the reflected signals
\vspace{-0.1cm}
\begin{equation}\label{eq:LaplaceReflect}
\begin{aligned}[b]
	&\mathcal{L}_{-Q_R}(s)  \\
	&= \bigg(\int_{0}^{2\pi} \mathcal{L}_{|\rho_{R}|^2}\!\Big(-s M^2 P_0 G(r, R, \psi)\Big) \frac{1}{2\pi} {\rm d}\psi \bigg)^N  \\
	&= \bigg(\!\frac{1}{2\pi} \!\!\int_{0}^{2\pi} \frac{1}{1 - s M^2 \beta^2P_0 R^{-2}|r+Re^{-\imath \psi}|^{-2}} {\rm d}\psi\bigg)^N  \\
	& = \!\bigg(\!\frac{1}{2\pi}\!\!\int_{0}^{2\pi}\! \!\bigg(\!1 \!+\!  \frac{\frac{M^2 \beta r^{\alpha_{\rm NLoS}}}{R^4}}{| \frac{r}{R}\!+\!e^{-\imath \psi}|^2\! -\! \frac{M^2 \beta r^{\alpha_{\rm NLoS}}}{R^4}}\!\bigg)\! {\rm d}\psi\!\bigg)^N \! \!\stackrel{(a)}{=}\!\Bigg(\! 1+\\
	&  \frac{ \frac{M^2 \beta r^{\alpha_{\rm NLoS}}}{R^4}}{\sqrt{\!\big(\frac{r}{R}\big)^4\! +\! \big( \frac{M^2 \beta r^{\alpha_{\rm NLoS}}}{R^4}\!-\!1\big)^2\! -\! 2 \big(\frac{r}{R}\big)^2 \big(1\! +\! \frac{M^2 \beta r^{\alpha_{\rm NLoS}}}{R^4}\big)}}\! \Bigg)^N,
\end{aligned}
\end{equation}
where $(a)$ follows from the same steps as in Eq.~\eqref{eq:elimination}.

For the RIS-UE link experiencing Nakagami-$m$ fading, {\color{black}where the fading power follows a Gamma distribution with shape $m$ and scale $\frac{1}{m}$, the Laplace transform of the fading power $\mathcal{L}_{\rho_{R}^2}(s)=\big(\frac{1}{1+s/m}\big)^m$.
By substituting the Laplace transform for Rayleigh fading in Eq~\eqref{eq:elimination} with that of Nakagami-$m$ fading, the Laplace transform $\mathcal{L}_{-Q_R}(s)$ is expressed as}
\begin{equation}
\begin{aligned}[b]
	&\mathcal{L}_{-Q_R}(s) = \\ &\bigg(\frac{1}{2\pi} \int_{0}^{2\pi}  \bigg(\frac{1}{1 - sM^2/m \beta^2P_0 R^{-2}|r+Re^{-\imath \psi}|^{-2}}\bigg)^m{\rm d}\psi\bigg)^N.
\end{aligned}
\end{equation}
This expression can be reduced to a form similar to  Eq.~\eqref{eq:elimination} and subsequently transformed into closed form. 
For example, for $m=2$,  the mapping, analogous to Eq.~\eqref{eq:elimination} is given by
\begin{equation}
	\begin{aligned}[b]
		&\frac{1}{2\pi}\int_0^{2\pi} \bigg(1+\frac{b}{-b+|a+e^{-\imath \psi}|^2}\bigg)^2{\rm d}\psi \\
		=& 1+\frac{2 b(1-a^2)^2 - 3 b^2(1+a^2)  + b^3}{\left(a^4-2 a^2 (b+1)+(b-1)^2\right)^{3/2}}.
	\end{aligned}
	\end{equation}
Consequently, we have the Laplace transform for the reflected links experiencing Nakagami-2 fading
\begin{equation}
\begin{aligned}[b]
			&\mathcal{L}_{-Q_R}\Big(\frac{1}{P_0g(r)}\Big)  = \Bigg(1+\\
			&\frac{\!\frac{M\!^2\!\beta r^{\alpha_{\rm NLoS}}}{R^4}(\!1\!-\!\big(\!\frac{r}{R}\!\big)\!^2\!)^2\! -\! 3\!\big(\!\frac{M\!^2\!\beta r^{\alpha_{\rm NLoS}}}{2R^4}\!\big)\!^2\!(1\!+\!\big(\!\frac{r}{R}\!\big)\!^2) \! +\! \big(\!\frac{M\!^2\!\beta r^{\alpha_{\rm NLoS}}}{2R^4}\!\big)\!^3}{\!\big(\!\big(\!\frac{r}{R}\big)^4 \!+ \!\big( \frac{M^2\beta r^{\alpha_{\rm NLoS}}}{2R^4}\!-\!1\big)^2\! -\! 2 \big(\frac{r}{R}\big)^2 \big(1 + \frac{M^2\beta r^{\alpha_{\rm NLoS}}}{2R^4}\!\big) \!\big)^{1.5}} \!\!\Bigg)^N\!\!.
\end{aligned}
\end{equation}
It is observed that this leads to a pattern of closed-form expressions for integer $m$, potentially provable by induction, which will be explored in future research.

\section{Numerical Analysis}\label{sec:simulation}

In this section, we first provide a cell-level validation for the proximity guideline of RIS placement, which establishes our circle-deployment stochastic geometry model. 
Using this model, we then analyze the impact of RIS-induced reflected interference, different channel fading assumptions, and alternative RIS placement models.
Numerical simulations are performed using both analytical expression evaluation (labeled ``analytical") and Monte Carlo methods (labeled ``simulation").
Finally, we compare their computational efficiency.

The simulations are configured with an average of 5 RIS panels per cell. 
The BS density is set to $10$ per ${\rm km}^2$, resulting in an approximate distance of $150$ meters between a randomly located UE to its nearest BS.
We evaluate the performance of a tagged UE positioned at a fixed distance of $r=100$ meters to its associated BS.
For the environment, we assume a Gaussian noise level of $\sigma^2=10^{-13}$ Watt.
The path loss exponents are set to $\alpha_{\rm LoS}=2$ and $\alpha_{\rm NLoS}=3$. 
Additional simulation parameters are specified for specific evaluations.

\subsection{Proximity deployment guideline}\label{subsec:proximity}

\vspace{-0,3cm}
\begin{figure}[htp!]
	\centering
	\includegraphics[width=1.0\linewidth]{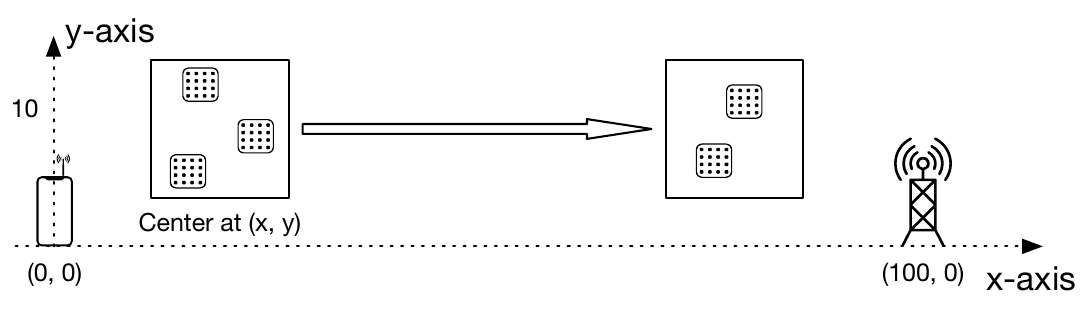}
	\caption{\color{black}The square area where multiple RISs are randomly deployed is varied in proximity from the UE to the BS.}
	\label{fig:proximity}
	\vspace{-0.3cm}
\end{figure} 
The proposed circle-deployment RIS model is motivated by the observation that RIS placement near the BS or UE minimizes reflected path loss, thereby increasing the power of reflected signals.
While previous studies have examined this deployment guideline in deterministic or one-dimensional settings, we investigate a stochastic deployment of multiple RISs within a two-dimensional square region.
As shown in Fig~\ref{fig:proximity}, we define the UE at $(0, 0)$ and the BS at $(100, 0)$, with the x-axis connecting them and the y-axis perpendicular, and set the square's width to 10 meters.
We vary the square's center along the x-axis, keeping the y-coordinates at 0 or 10 as labeled in the figure, transitioning from the UE's proximity to the BS's proximity, as illustrated in Fig.~\ref{fig:near_fig}.
RISs within the square are modeled as either a BPP with 5 RISs per square, or a PPP with an average of 5 RISs per square. 
Each RIS comprises $M=80$ elements for beamforming.
Interfering BSs outside the serving cell are modeled as the PPP, and no RISs are associated with these BSs for this simulation.

\begin{figure}
	\centering
	\includegraphics[width=0.9\linewidth]{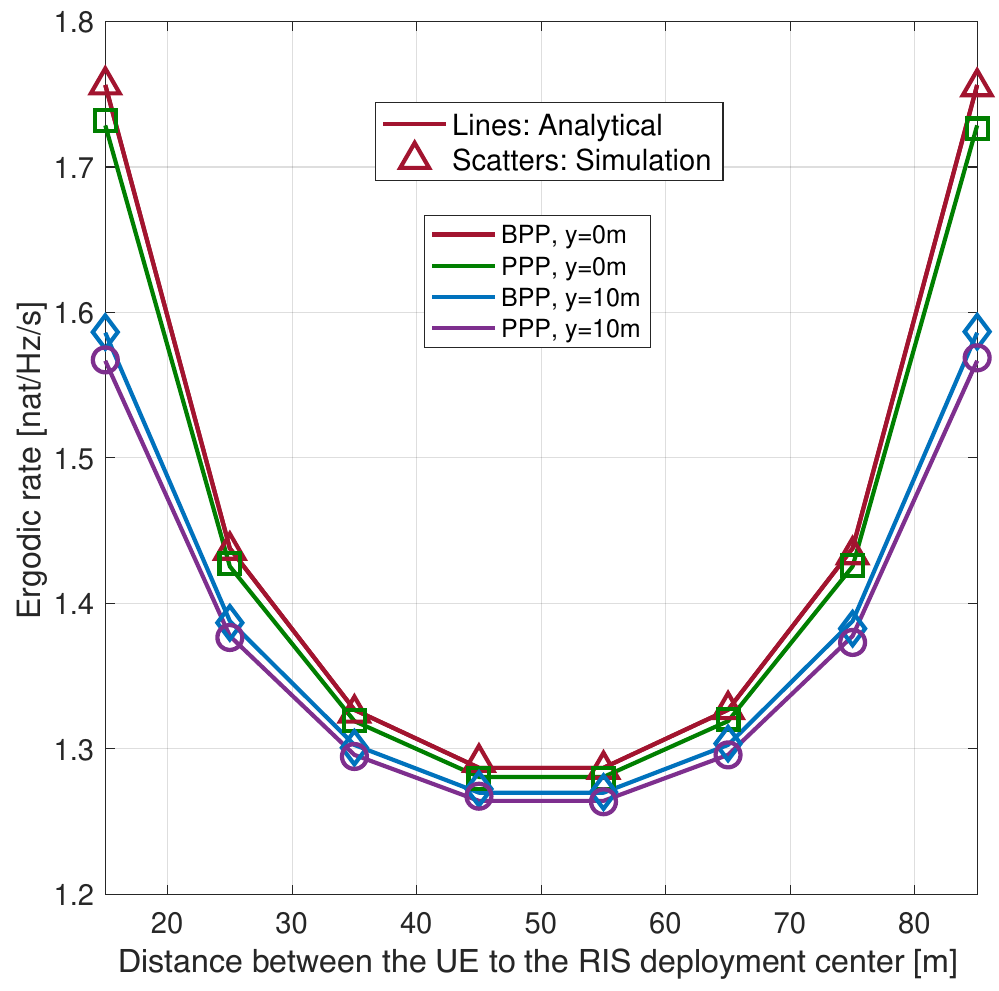}
	\caption{\color{black}The relationship between the RIS-enhanced ergodic rate and the distance of the RIS deployment area to the BS.}
	\label{fig:near_fig}
	\vspace{-0.5cm}
\end{figure}

As shown in Fig.~\ref{fig:near_fig}, the ergodic rate improvement provided by RISs is maximized when the square region's center is close to either the BS or the UE.
Specifically, at x-coordinates of approximately 15 (RISs near the UE) and 85 (RISs near the BS), the rate enhancement for $y=0$ is roughly $50\%$ higher compared to the scenario where the square is positioned in the middle between the BS and UE. 
Furthermore, when we analyze the y-coordinate, the set of curves with y-coordinate 0 is above that with y-coordinate 10, further validating the proximity deployment guideline.
This observation applies to both the BPP and PPP models.
We can conclude that, in a two-dimensional stochastic deployment, RISs should be placed in proximity to either the BS or the UE to effectively reflect the received signal.

\subsection{System level RIS reflected interference}
We now focus on the circle-shaped RIS deployment model, with RISs deployed circularly around each BS, for which we configure the radius of the circle $R=20$m. 
Recall that all RISs beamform and reflect signals from their associated BSs, and that reflected beams may overlap the tagged UE. 
We thus conduct a system-level analysis of RIS reflection interference. 
We assume that the RIS-reflected beams experience Rayleigh fading.
In Fig.~\ref{fig:interference}, we investigate the reflection interference by studying two key factors: the cross-cell path loss exponent, labeled by $\alpha_{IR}$, which can refer either LoS $\alpha_{\rm LoS}=2$ or NLoS $\alpha_{\rm NLoS}=3$, and the overlapping probability, $\xi$, of interfering beam. 
Specifically,  $\xi$ can be fixed to a constant (0.1 or 0.3) or inversely proportional to the number of RIS elements ($1/M$ or $5/M$), depending on the RIS beamforming model.

\begin{figure}
\centering
\includegraphics[width=0.9\linewidth]{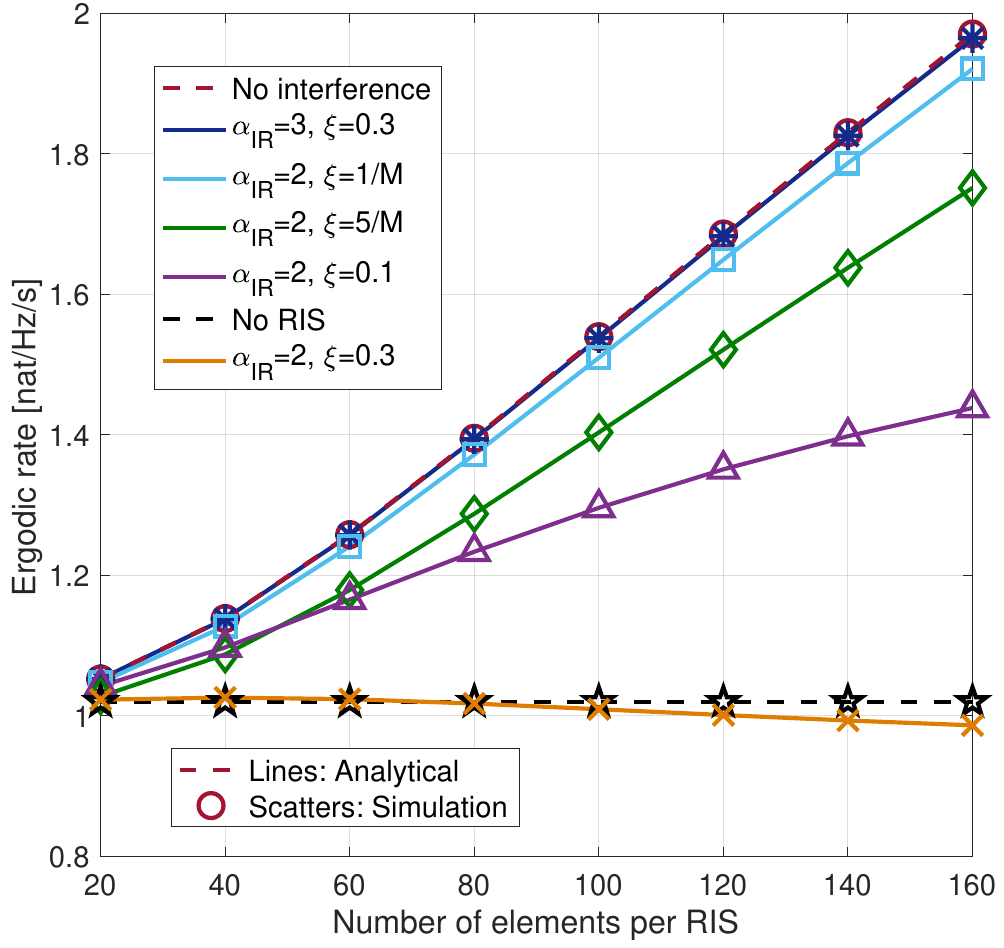}
\caption{\color{black}Analysis of the impact of RIS-generated inter-cell interference on the ergodic rate. }
	\label{fig:interference}
	\vspace{-0.5cm}
\end{figure}

Fig.~\ref{fig:interference} shows the ergodic rate of the tagged UE as a function of the number of RIS elements.
Specifically, we present two reference curves: the black dashed line with pentagram markers represents the scenario without RIS deployment, and the red dashed line with circle marks represents the scenario without RIS reflection interference.
Since cross-cell signal propagation usually occurs through an NLoS channel ($\alpha_{IR}=3$), the resulting ergodic rate (the blue curve with star marks) aligns with the reference curve that excludes RIS reflection interference. 
This suggests that RIS reflection interference has a minimal effect under high cross-cell propagation loss.
When the RIS reflection interference experiences LoS channel ($\alpha_{IR}=2$),  the reflected interference becomes significant.
For overlapping probabilities of $\xi=1/M$ or $\xi=5/M$, as shown by the curves with square and diamond markers, a noticeable ergodic rate penalty is observed. 
The impact of reflection interference is more significant for fixed overlapping probabilities ($\xi=0.1$ and $\xi=0.3$) when the number of elements per RIS is large, as shown by the curves with triangle and cross marks.
In particular for $\xi=0.3$, with a limited number of RIS elements, the deployment provides only a marginal ergodic rate gain compared to the no-RIS baseline. 
As the number of elements per RIS increases, performance degrades below this baseline.
However, in practical scenarios, strong cross-channel conditions are uncommon, thus such extreme performance degradation are unlikely. 
Moreover, with large RIS deployments, the reflected beam is typically narrow, resulting in a low overlapping probability. 
Therefore, we can reasonably assume reflection interference to be negligible in the following analysis.

\subsection{Investigation of fading models }
This work's closed-form derivation relies on the simplifying assumption of Rayleigh fading for the reflected link. 
However, accurately modeling the RIS-reflected beam's variation in practice is complex for its dependence on beamforming.
To quantify the potential uncertainty introduced by fading models, we evaluate the ergodic rate of the tagged UE under Nakagami-$m$ fading distributions, ensuring a constant reflected signal power.  
Furthermore, utilizing a large number of RIS elements for beamforming can induce channel hardening, thereby mitigating fading effects. 
Therefore, we also examine the scenario with a constant power reflected beam.

\begin{figure}
	\centering
	\includegraphics[width=0.9\linewidth]{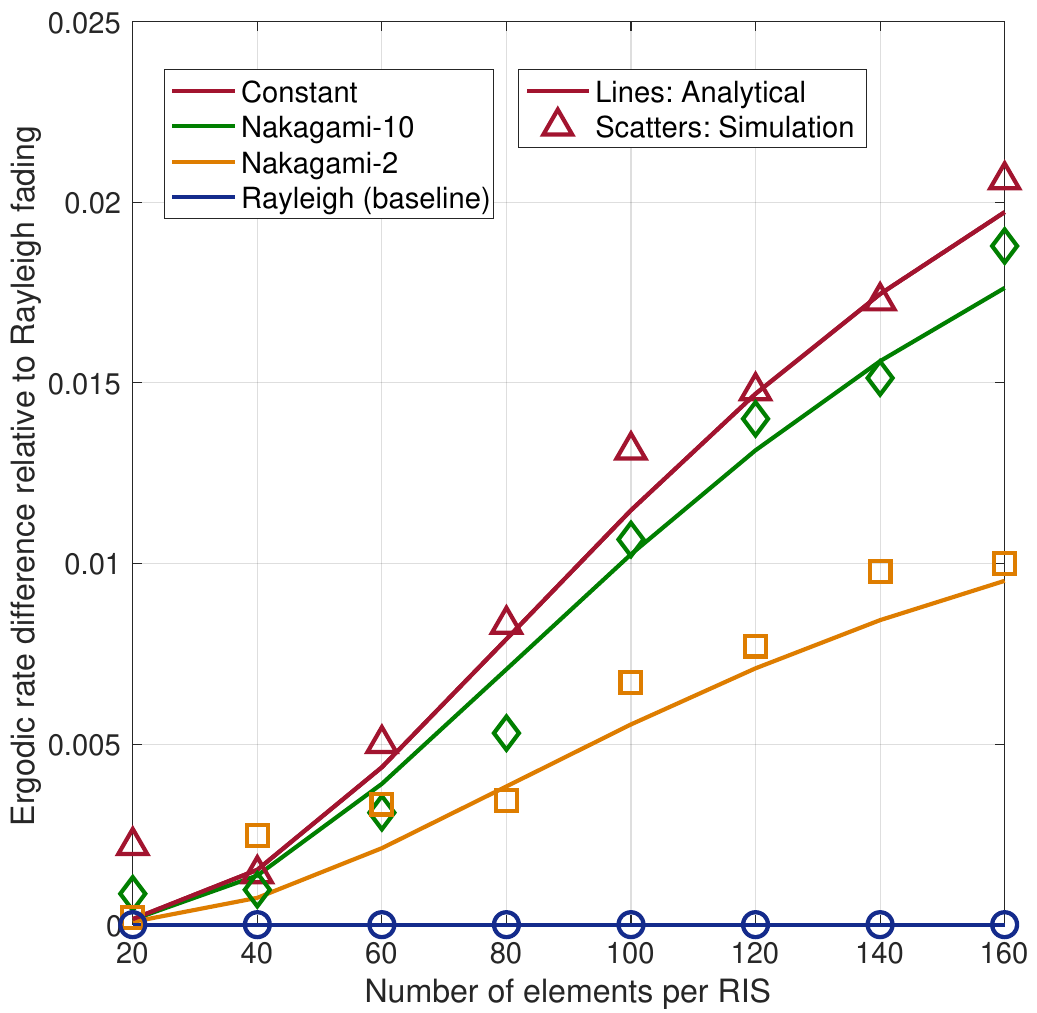}
	\caption{Ergodic rate sensitivity of RIS reflection links under different fading models, compared to Rayleigh fading.}
	\label{fig:fading}
	\vspace{-0.5cm}
\end{figure}

Fig.~\ref{fig:fading} compares the relative difference of ergodic rate of different fading models (specified in the legend) to the Rayleigh fading baseline. 
Although the deviation increases with the number of RIS elements, its magnitude remains on the order of $10^{-2}$, indicating a negligible impact on system performance.
This analysis demonstrates how fading variations influence the approximated ergodic rate.
For instance, within the Nakagami-$m$ fading model, larger $m$ values represent reduced variation and correlate with slightly higher ergodic rates. 
Furthermore, the constant reflected power model, which eliminates all variation, gives the highest ergodic rate. 
Therefore, simplified fading models can bound the ergodic rate for unknown fading conditions.
Specifically, the Rayleigh model establishes the lower bound, while the constant power model can provide an upper bound.
Consequently, for system-level performance assessment, simple models like Rayleigh and constant power are sufficient. 
This justification further validates the use of the derived closed-form expressions for RIS-reflected signal power, especially those based on Rayleigh fading.

\subsection{RIS placement and selection}

While the analytical results presented rely on PPP/BPP RIS deployment models, practical scenarios often exhibit repulsiveness between RIS nodes. 
Therefore, we evaluate the ergodic rate for equidistant RIS deployment along a circular arc. 
Additionally, we evaluate the case where only the strongest-link RIS is selected from a BPP-distributed set. 
As these two scenarios lack analytical solutions, we utilize Monte Carlo simulations for evaluation.

\begin{figure}
\centering
\includegraphics[width=0.9\linewidth]{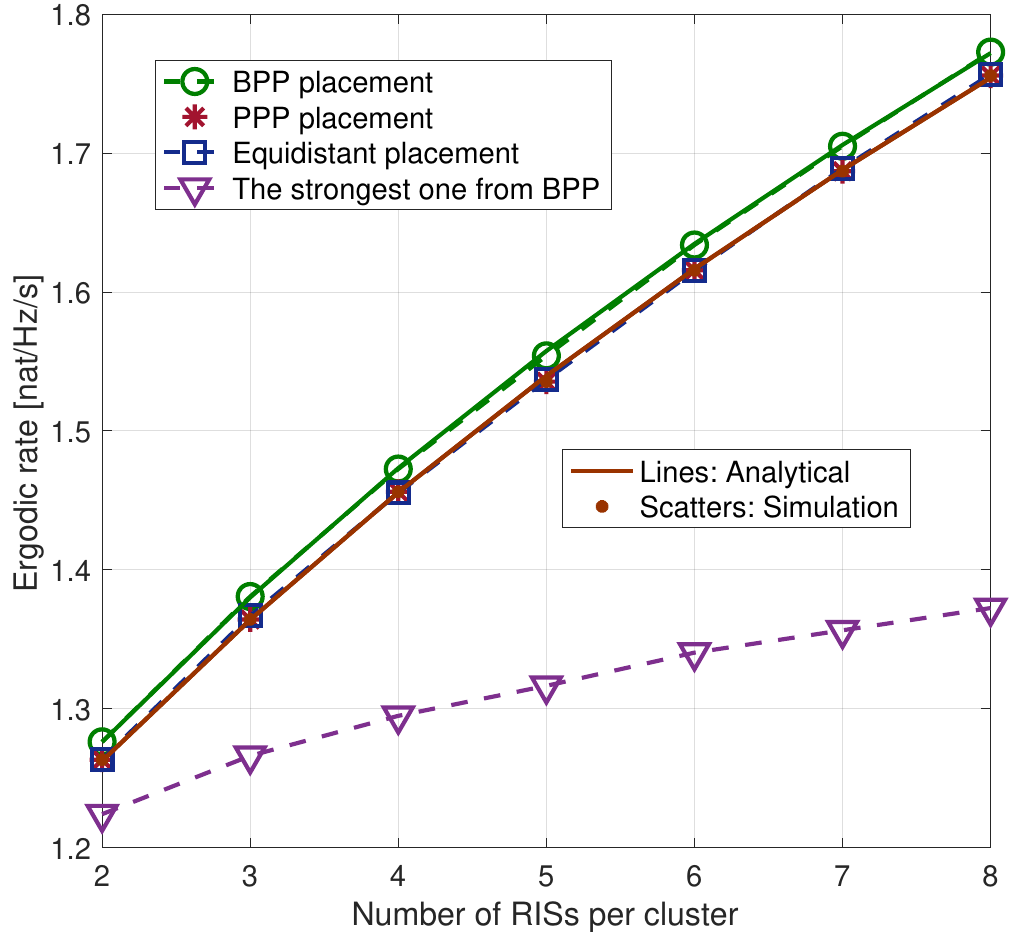}
\caption{Ergodic rate analysis for four RIS location scenarios: BPP, PPP, equidistant deployment, or best-link selection from a BPP.}
\label{fig:RIS_PPs}
\vspace{-0.5cm}
\end{figure}
In Fig.~\ref{fig:RIS_PPs}, the BPP model demonstrates a marginally higher ergodic rate compared to the PPP model, consistent with the observations in  Fig.~\ref{fig:near_fig}.
This suggests that varying the number of RISs can lead to a slight reduction in performance enhancement, given a fixed average RIS resource.
Utilizing multiple RISs to serve a UE is preferred because selecting only the optimal RIS for reflection results in a reduced performance gain with respect to the number of RISs, while also incurring additional signaling overhead.
Furthermore, the performance curve for an equidistant RIS deployment closely aligns with that of randomly located RISs. 
This indicates that the analytical model based on BPP/PPP accurately characterizes the system performance.

\subsection{Performance evaluation efficiency}
We conclude by comparing the computational efficiency of the analytical and Monte Carlo methods.
For all simulations described in this section, the numerical evaluation times\footnote{We performed both simulations on a 2 GHz Quad-Core Intel Core i5 laptop for a fair comparison.} are summarized in Table~\ref{tab:time_comparison}. 
The reported time value represents the average execution time required to compute a single data point shown in the referenced figure.
\begin{table}[htp!]
\centering
\vspace{-0.4cm}
\caption{Simulation time}
\label{tab:time_comparison}
\begin{tabular}{|c|c|c|}
\hline
&Analytic computation & Monte Carlo\\
\hline
Fig.~\ref{fig:near_fig}& 2.01s & 21.5s \\
\hline
Fig.~\ref{fig:interference}& 1.49s & 191s \\
\hline
Fig.~\ref{fig:fading}& 22.3s & 246.1s \\
\hline
Fig.~\ref{fig:RIS_PPs}& 0.101s & 12.5s \\
\hline
\end{tabular}
\vspace{-0.3cm}
\end{table}

In Table~\ref{tab:time_comparison}, it is shown that the analytical method offers significantly higher time efficiency in evaluating system performance for a given simulation configuration. 
Furthermore, the analytical results exhibit less variation, especially when the values are small, as shown in Fig.~\ref{fig:fading}. 
This implies the Monte Carlo method would necessitate a considerably higher number of simulation runs for comparable accuracy. 
Therefore, considering both time efficiency and accuracy, the analytical method outperforms the Monte Carlo approach. 

\section{Conclusions and Future Works}\label{sec:conclusion}
This work presents an efficient stochastic geometry analysis, characterizing reflected signals from multiple RISs using their Laplace transform and deriving closed-form expressions for specific scenarios.
Based on this analysis, we identify conditions under which the RIS-reflected interference can be neglected or becomes detrimental.
We show that a Rayleigh fading model accurately represents RIS reflected links, given the reflected power. 
Furthermore, we show that various RIS placement strategies can be readily adapted and evaluated, validating the suitability of BPP and PPP models for multiple RIS deployments.  
Finally, our analytical approach provides a significant advantage in both computational efficiency and accuracy over Monte Carlo simulations.

The availability of closed-form expressions for the reflected signal power enables efficient analysis and expands the scope of potential applications. 
For example, incorporating stochastic geometry models of random blockage in future research will enhance our understanding of RIS deployment in realistic environments.

\vspace{-0.2cm}
\section{Acknowledgement}
The work of G. Sun and F. Baccelli was supported by the European
Research Council project titled NEMO under grant ERC
788851 and by the CIFRE Ph.D. program, under the grant agreement number 2021/0304 to NOKIA Networks France. 
The work of F. Baccelli was also supported by the Horizon Europe project titled INSTINCT under grant SNS 101139161, and by the French National Agency for
Research project titled France 2030 PEPR réseaux du Futur
under grant ANR-22-PEFT-0010.
\bibliographystyle{ieeetr}
\bibliography{reference.bib}

\end{document}